\documentclass[letterpaper, 10pt, conference]{ieeeconf}

\IEEEoverridecommandlockouts

\makeatletter

\let\proof\@undefined
\let\endproof\@undefined
\makeatother

\usepackage{amsmath,amsfonts,amsthm, amssymb,color, graphicx}
\usepackage{caption}
\usepackage{subcaption}

\usepackage{algorithmic,algorithm}
\usepackage[algo2e,vlined,ruled,linesnumbered]{algorithm2e}

\graphicspath{{./fig/}}

\newtheorem{theorem}{Theorem}[section]

\newtheorem{corollary}[theorem]{Corollary}
\newtheorem{lemma}[theorem]{Lemma}

\theoremstyle{definition}
\newtheorem{definition}[theorem]{Definition}

\newtheorem{example}[theorem]{Example}
\theoremstyle{remark}
\newtheorem{remark}[theorem]{Remark}

\usepackage{epsfig} 
\usepackage{psfrag}
\usepackage{mathptmx}      
\usepackage{helvet}     
\usepackage{courier}     
\usepackage{type1cm}      

\usepackage{makeidx}         
\usepackage{graphicx}        
                           
\usepackage{multicol}        
\usepackage[bottom]{footmisc}
\usepackage{algorithmic,algorithm}
\usepackage{amssymb}
\usepackage{epsfig,color}

\newcommand{\etal}{{\emph{et al.}}}

\title{Distributed Dominating Sets on Grids}
\author{Elaheh Fata \qquad Stephen L. Smith \qquad Shreyas Sundaram
\thanks{This research is partially
    supported by the Natural Sciences and Engineering Research Council
    of Canada (NSERC). } 
\thanks{The authors are with the Department of Electrical and
  Computer Engineering, University of Waterloo, Waterloo ON, N2L 3G1
  Canada. (email: efata@uwaterloo.ca;
  stephen.smith@uwaterloo.ca;  ssundara@uwaterloo.ca) }
}

\begin{document}

\maketitle
\thispagestyle{empty}
\pagestyle{empty}

\begin{abstract}
This paper presents a distributed algorithm for finding near optimal dominating sets on grids. The basis for this algorithm is an existing centralized algorithm that constructs dominating sets on grids. The size of the dominating set provided by this centralized algorithm is upper-bounded by $\left\lceil\frac{(m+2)(n+2)}{5}\right\rceil$ for $m\times n$ grids and its difference from the optimal domination number of the grid is upper-bounded by five. 
Both the centralized and distributed algorithms are generalized for the $k$-distance dominating set problem, where all grid vertices are within distance $k$ of the vertices in the dominating set.
\end{abstract}

\section{Introduction}
\label{sec:intro}

Significant attention has been devoted in recent years to the study of large-scale sensor and robotic networks due to their promise in fields such as environmental monitoring~\cite{NEL-DP-FL-RS-DMF-RD:07}, inventory warehousing~\cite{Guizzo08}, and reconnaissance~\cite{RWB-TWM-MAG-EPA:02}.  One of the key objectives in such networks is to ensure {\it coverage} of a given area, where every point in the space is within the sensing radius of one or more of the agents (i.e., sensors or robots).  Various algorithms have been proposed to achieve coverage based on differing assumptions on the mobility and sensing capabilities of the agents \cite{DistCtrlRobotNetw}.

In certain scenarios, the environment may impose restrictions on the feasible locations and motion of the agents \cite{SML:06}.  In such cases, it is natural to model the environment as a {\it graph}, where each node represents a feasible location for an agent, and edges between nodes indicate available paths for the agents to follow.  The coverage capabilities of any given agent are then related to the shortest-path distance metric on the graph: an agent located on a node can cover all nodes within a certain distance of that node.  The goal of selecting certain nodes in a graph so that all other nodes are within a specified distance of the selected nodes is classical in graph theory, and is known as the {\it dominating set problem} \cite{GJ'79}.  Versions of this problem appear in settings such as multi-agent security and pursuit \cite{Abbas12}, routing in communication networks \cite{Wu01}, and sensor placement in power networks \cite{Dorfling06}.

Finding the {\it domination number} (i.e., the size of a smallest dominating set) of arbitrary graphs is NP-hard \cite{GJ'79}. In fact, Raz and Safra showed that achieving an approximation ratio better than $c\log n$ for the dominating set problem in general graphs is NP-hard, where $c>0$ is some constant and $n$ is the number of vertices of the graph~\cite{RazSaf'97}. 
However, there are several algorithms known for the dominating set problem for which the ratio between the size of the resulting dominating set and graph domination number can closely reach the $c\log n$ bound. The simplest of these algorithms is a \emph{greedy algorithm} that at each step adds one vertex to the dominating set. The vertices that are already in the dominating set are marked as `black', the vertices that share edges with black vertices are marked as `gray' and other vertices are `white'. At each step, a white vertex that shares the maximum number of edges with other white vertices are added to the dominating set and the colour labels of all vertices are updated according to the aforementioned rules. Another widely used approximation algorithm for the problem uses a \emph{linear programming relaxation}. Both greedy and linear programming approaches for the dominating set problem are known to have $(\ln n+1)$-approximation ratios~\cite{Johnson'74,lovasz'75}, which are in $O(\log n)$.

Even though in general graphs one cannot obtain an approximation ratio in $o(\log n)$, in special types of graphs better approximation ratios are obtainable. One of the most important classes of graphs are \emph{planar graphs}. A \emph{planar graph} is a graph that can be drawn in a plane so that none of its edges intersect except at their ends~\cite{BM'08}. The dominating set problem is still NP-hard for planar graphs; however, the domination number of this type of graphs can be approximated within a factor of $(1+\epsilon)$ for an arbitrarily small $\epsilon>0$~\cite{Baker'94}. 

Grid graphs are a special class of graphs that have attracted attention due to their ability to model and discretize rectangular environments~\cite{LBTD'03,LJDKM'00}.
Grids can be  used in simplifying the underlying environment and limiting energy consumption by representing a certain area of the environment with only one node in the grid~\cite{BJDMTACX'05}. Moreover, grid graphs, due to their special structure that do not leave any area of environment unrepresented while transferring the problem environment into a tractable domain, can successfully provide efficient area coverage and hence are used very commonly in the network coverage and delectability literature~\cite{LBTD'03,CW'04}. All these application motivated us in studying the dominating set problem when the underlying graph is a grid.

As discussed above, it is NP-hard to find the domination number of general or even planar graphs. It can be easily observed that grid graphs lie in the class of planar graphs and hence their domination number can be obtained within a small ratio. However, due to the special structure of grids, their domination number can in fact be determined optimally, although the path to obtaining the exact domination number of grids was not straightforward. For $m \times n$ grid graphs, the size of the optimal dominating set was unknown until recently, although an upper bound of $\left\lfloor \frac{(m+2)(n+2)}{5}\right\rfloor-4$ was shown in \cite{TYC'92} for $8\le m\le n$ using a constructive method.  Various attempts have been made in recent years to find a tight lower bound on the size of the optimal dominating set.  In \cite{ACIP'11}, the authors used brute-force computational techniques to find optimal dominating sets in grids of size up to $n = m = 29$.  The paper \cite{GPRT'11} showed finally that the lower bound on the domination number is equal to the upper bound for~$16\le m\le n$, thus characterizing the domination number in grids.  

In this paper, we make two contributions to the study of dominating sets on grids, and their application to multi-agent coverage. First, we provide a distributed algorithm that locates a set of agents on the vertices of an $m\times n$ grid such that they construct a dominating set for the grid and the required number of agents is within a constant error from the optimal. 
The agents require only limited memory, sensing and communication abilities, and thus the solution is applicable to multi-robot coverage applications where the environment can be discretized as a grid. Our distributed algorithm is based on a simple constructive method to obtain {\it near-optimal} dominating sets (i.e., that require no more than $5$ vertices over the optimal number) in grids by Chang $\etal$~\cite{TYC'92}. This  approach is based on a systematic tiling pattern that we call a {\it diagonalization}. Second, we generalize Chang's construction to the $k$-distance dominating set problem, where a given vertex can cover all other vertices within a distance $k$ from it. We show that our distributed algorithm can also be generalized to work in the $k$-distance domination scenario.

In Section \ref{sec:background}, we introduce the essential models and notation for formulating the dominating set problem on grids. In Section \ref{sec:center-domin} we discuss the constructive centralized grid domination algorithm. The materials in Section \ref{sec:center-domin} are used in Section \ref{sec:dist-domin} to design a distributed algorithm for the dominating set problem. Section \ref{sec:k-distance} generalizes the results in Sections \ref{sec:center-domin} and \ref{sec:dist-domin} to the $k$-distance dominating set problem. Finally, Section \ref{sec:conclusion} concludes the paper and discusses the corresponding open problems.

\section{Background}
\label{sec:background}
A graph $G=(V,E)$ is defined as a set of vertices $V$ connected by a set of edges $E\subseteq V\times V$. We assume the graph is \emph{undirected}, i.e., $(v,u)\in E\Leftrightarrow (u,v)\in E,\forall v,u\in V$. A vertex $u\in V$ is defined as a \emph{neighbour} of vertex $v\in V$, if $(u,v)\in E$. The set of all neighbours of vertex $v$ is denoted by $N(v)$. For a set of vertices $U\subseteq V$, we define $N(U)$ as $\bigcup_{u\in U}{N(u)}$. For a set of vertices $U\subseteq V$, we say the vertices in $N(U)$ are \emph{dominated} by the vertices in $U$. For graph $G$, a set of vertices $S\subseteq V$ is a \emph{dominating set} if each vertex $v\in V$ is either in $S$ or is dominated by $S$.

A dominating set with minimum cardinality is called an \emph{optimal dominating set} of a graph $G$; its cardinality is called the \emph{domination number} of $G$ and is denoted by $\gamma(G)$. Note that although the domination number of a graph, $\gamma(G)$, is unique, there may be different optimal dominating sets \cite{CLRS'01}.

Here, we study the dominating set problem on a special class of graphs called \emph{grid graphs}. An $m\times n$ grid graph $G=(V,E)$ is defined as a graph with vertex set $V=\{ v_{i,j}| 1\le i\le m, 1\le j\le n\}$ and edge set $E=\{(v_{i,j},v_{i,j'})|~|j-j'|=1\}\bigcup \{(v_{i,j},v_{i',j})|~|i-i'|=1\}$ \cite{BM'08}. For ease of exposition, we will fix an orientation and labelling of the vertices, so that vertex $v_{1,1}$ is the lower-left vertex and vertex $v_{m,n}$ is the upper-right vertex of the grid.
We denote the domination number of an $m\times n$ grid $G$ by $\gamma_{m,n}=\gamma(G)$.

\begin{theorem}[Gon\c{c}alves $etal$,~\cite{GPRT'11}]
\label{theo:cited-4}
For an $m\times n$ grid with $16\le m\le n$, $\gamma_{m,n}=\left\lfloor \frac{(m+2)(n+2)}{5}\right\rfloor-4$.
\end{theorem}

Our distributed grid domination algorithm is based on a procedure developed by Chang $\etal$~\cite{TYC'92}. To obtain the tools needed in our distributed algorithm, we discuss an overview of Chang's algorithm in Section~\ref{sec:center-domin}. These tools are used in Sections~\ref{sec:dist-domin} and~\ref{sec:k-distance} in the distributed domination algorithm and in developing the $k$-distance dominating set results. Before discussing these tools, we introduce two useful definitions.

\begin{definition}(Grid Boundary)
\label{def:boundaries}
For an $m\times n$ grid $G=(V,E)$, we define the \emph{boundary} of $G$, denoted by $B(G)$, as the set of vertices with less than 4 neighbours.
\end{definition}

\begin{definition}(Sub-Grids and Super-Grids)
\label{def:sub-super-grids}
An $m\times n$ grid $G=(V,E)$ is called a \emph{sub-grid} of an $m'\times n'$ grid $G'=(V',E')$ if $G$ is induced by vertices $v'_{i,j}\in V'$, where $2\le i\le m'-1$ and $2\le j\le n'-1$. If $G$ is a sub-grid of $G'$, $G'$ is called the \emph{super-grid} of $G$ (see Figure \ref{fig:eg}(a)).
\end{definition}

\section{Overview of Centralized Grid Domination Algorithm}
\label{sec:center-domin}
In \cite{ACIP'11}, Alanko $\etal$ provided examples of optimal dominating sets for $n\times n$ grids with $1\le n\le 29$, obtained via a brute-force computational method. A visual inspection of these examples shows that as the size of the grid increases, the patterns of dominating vertices become more regular in the interior of grids, with irregularities at the boundaries. Figure \ref{fig:patterns} demonstrates some examples of patterns that arise in the dominated grids in \cite{ACIP'11}.

\begin{figure}[t]
\includegraphics[scale=0.25]{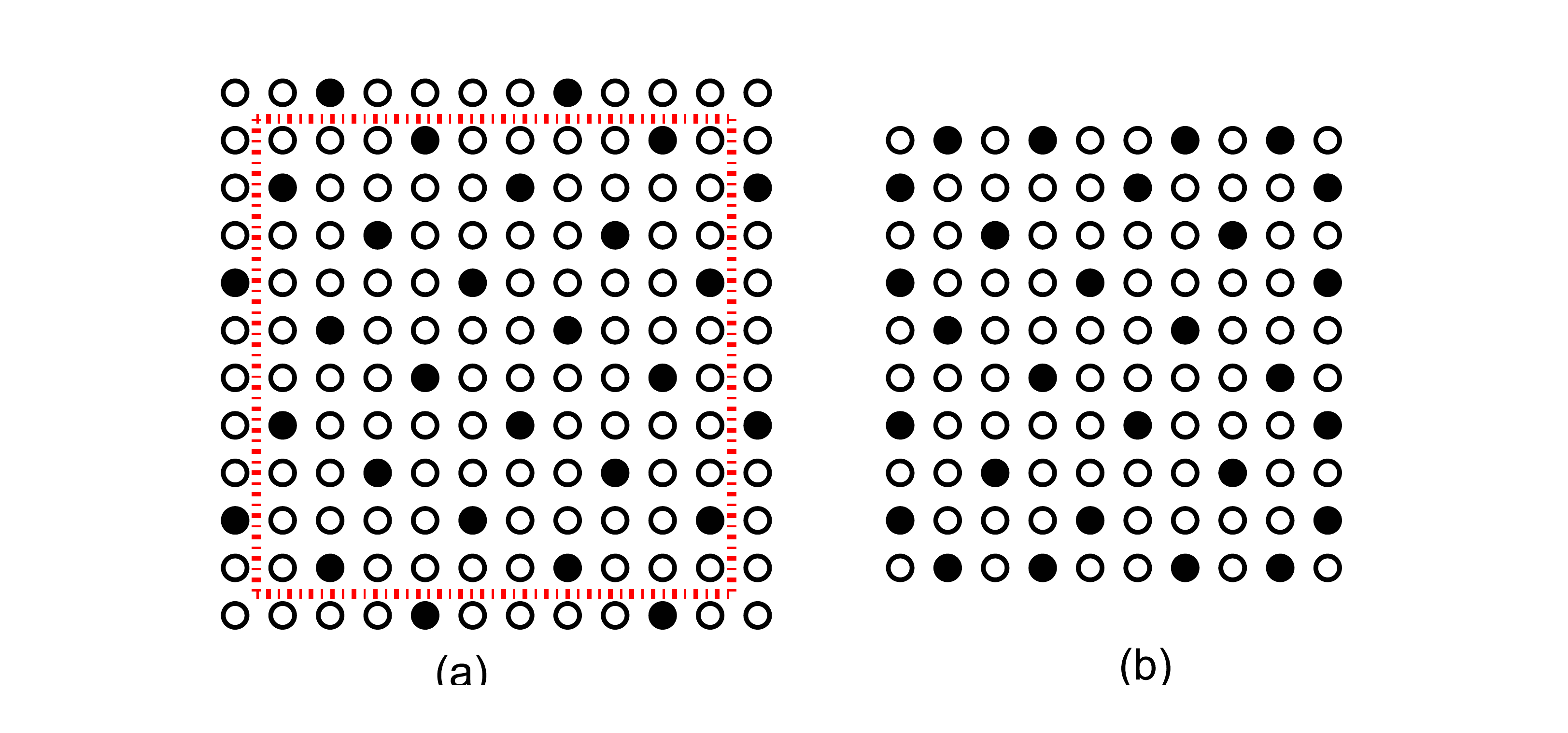}
\caption{In (a), a $12\times 12$ grid $G'$ is demonstrated and its $10\times 10$ sub-grid $G$ is highlighted by a red dashed square. $G'$ is diagonalized by a set $U'$ of 28 vertices. In (b), vertices in $U'\backslash V$ are projected onto their neighbours in $G$.}
\label{fig:eg}      
\end{figure}

Among the patterns used to dominate grids, the one illustrated in Figure \ref{fig:patterns}(b) is the most efficient, since there is no vertex that is dominated by more than one dominating vertex in this pattern. Hence, this pattern would be useful in obtaining dominating sets with near optimal size. We refer to the structure in Figure \ref{fig:patterns}(b) as a \emph{diagonal pattern}. Chang $\etal$ in~\cite{TYC'92} used these patterns to provide an upper-bound on the domination number of grids. As the proof on the upper-bound they obtained was constructive, we could derive a centralized algorithm for finding near-optimal dominating sets from their constructions. 
In this section, we provide an overview of Chang's construction and the derived algorithm from their results which we will use in the subsequent sections. First, we define the diagonal patterns formally as follows.  
Note that the $x$ and $y$ axes are as shown in Figure~\ref{fig:patterns}.

\begin{definition}(Diagonal Pattern)
\label{def:diag-patt}
A set of vertices $U\subset V$ constitutes a \emph{diagonal pattern} on grid $G=(V,E)$ if there exists a fixed $r\in \{0,1,2,3,4\}$ such that for any vertex $v_{x,y}\in U$ we have $y-2x\equiv r \pmod{5}$.
\end{definition}

\begin{definition}(Diagonalization)
\label{def:diagonalization}
A set of vertices $U\subset V$ \emph{diagonalizes} grid $G=(V,E)$ if it constitutes a diagonal pattern and there exists no vertex $v\in V\backslash U$ that can be added to $U$ so that $U$ remains a diagonal pattern.
\end{definition}

An example of a diagonalization is shown in Figure \ref{fig:eg}(a).\footnote{One can also define a diagonal pattern as a set of vertices whose $(x,y)$ coordinates satisfy $x-2y\equiv r \pmod{5}$, for some fixed $r$. This corresponds to swapping the $x$ and $y$ axes. For the proofs we only analyze the case mentioned in Definition \ref{def:diag-patt}; the other case can be treated similarly.} 
The algorithm derived from Chang's construction consists of the following two main steps:
\begin{enumerate}
  \item Diagonalization: At this step, a set of vertices $U$ that diagonalizes the grid is provided. 
  \item Projection: Using a process called \emph{projection}, the vertices that were not dominated by vertices in $U$ are characterized and new vertices are added to $U$ to dominate those vertices as well.
\end{enumerate}

We know discuss these two steps in more details. Chang $\etal$ showed that if a grid $G=(V,E)$ is diagonalized by a set of vertices $U\subset V$, then for any vertex $v\in (V\backslash U)$ that is not located on the grid's boundary there exists exactly one vertex in $U$ that shares an edge with $v$. In other words, every node that is not located on the grid's boundary, $B(G)$, is dominated by exactly one vertex in $U$. Moreover, they proved that if a set of vertices $U\subset V$ diagonalizes an $m\times n$ grid $G=(V,E)$, then $U$ contains at most $\left\lceil \frac{mn}{5}\right\rceil$ vertices. To construct a dominating set for $G$ it only remains to add some vertices to $U$ so that the resulting set dominates the vertices on the boundary as well. The vertices located on $B(G)$ with no neighbour in $U$ are called \emph{orphans} and are defined formally as follows.

\begin{definition}(Orphans)
\label{def:orphan}
Let $U\subset V$ be a set of vertices that diagonalizes grid $G=(V,E)$. A vertex $v\in V$ that has no neighbour in $U$ is called an \emph{orphan} (see Figure \ref{fig:eg}(a)).
\end{definition}

To dominate orphans, Chang $\etal$ used the super-grid of $G$, denoted by $G'=(V',E')$. Since the vertices on the boundary of $G$ lie inside grid $G'$, a set of vertices $U'\subset V'$ that diagonalizes $G'$ dominates all vertices of $G$. Moreover, it can be easily seen that the set of vertices $U=U'\cap V$ is a diagonalization for grid $G$.

Recall that diagonalization results in every vertex being dominated by at most one vertex in the diagonal pattern. Therefore, if a set of vertices $U'\subset V'$ diagonalizes $G'=(V',E')$, that is, the super-grid of $G=(V,E)$, then there are vertices in $B(G)$ that are dominated by vertices in $U'\backslash V$. Hence, the \emph{orphan} of a vertex $v\in U'\backslash V$ is a vertex $u\in B(G)$ such that $u\in N(v)$, and is denoted by $u=\mathrm{orphan}(v)$.

\begin{corollary}
\label{cor:orphan-n+m}
For an $m\times n$ grid $G$, the number of orphans is $O(n+m)$.
\end{corollary}

Since by diagonalizing $G'$ the orphans in $G$, i.e, vertices in $N(U'\backslash U)\cap V$, are dominated by the dominating vertices on the boundary of $G'$, a procedure called \emph{projection} is introduced that projects the dominating vertices in $B(G')$ inside sub-grid $G$. Hence, projection results in having all vertices in $G$ being dominated. This procedure is defined formally as follows.

\begin{definition}(Projection)
\label{def:porjection}
Consider a grid $G=(V,E)$ and its super-grid $G'=(V',E')$. For a set $U'\subseteq V'$, its \emph{projection} is defined as
the set $U''=\big(N(U'\backslash V)\cup U'\big)\cap V$. 
Similarly, we say a vertex $v\in U'\backslash V$ is \emph{projected} if it is mapped to its neighbour in $V$.
\end{definition}

\begin{figure}[t]
\includegraphics[scale=0.3]{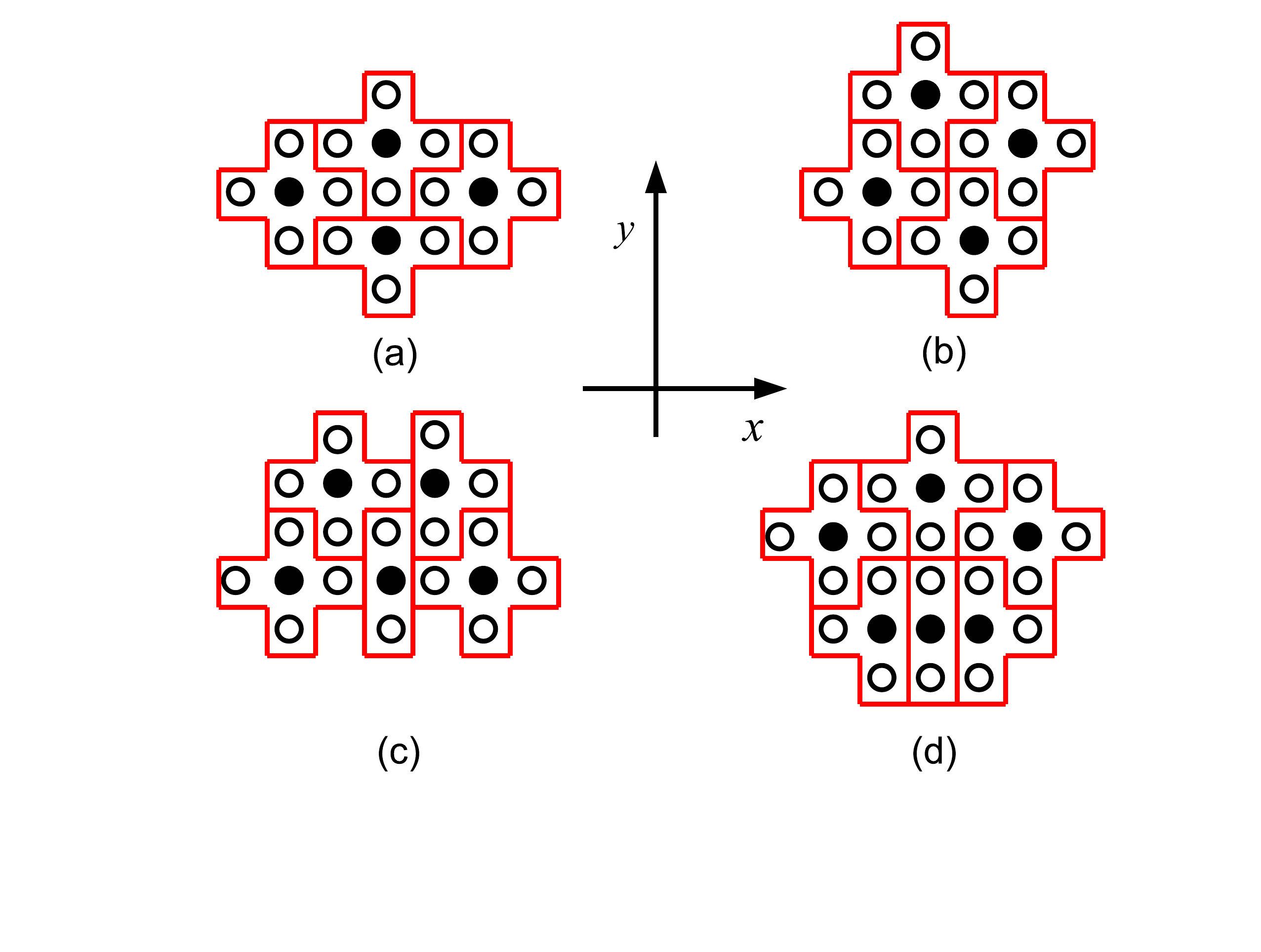}
\vspace{-1cm}
\caption{Examples of dominating vertex patterns
that appear in optimally dominated grids. The black vertices are the dominating vertices. The red line segments form regions so that in each region there exist one black vertex and at most four white vertices dominated by that black vertex.}
\label{fig:patterns}
\end{figure}

Figure \ref{fig:eg}(b) shows an example of a projection.
For grid $G=(V,E)$, its super-grid $G'=(V',E')$ and set $U'\subset V'$ that diagonalizes $G'$, by performing projection, the size of the obtained dominating set of $G$ is between $|U'|-4$ and $|U'|$. This is due to the fact that a vertex $v\in U'$ located at any corner of $G'$ has no neighbour in $V$ and hence, after projection it is not mapped into $V$. Since $G'$ has four corners, for $U''$, the result of projection of $U'$, we have $|U'|-4\le |U''|\le |U'|$. Hence $|U'|$, that is, the number of dominating vertices used in diagonalizing the super-grid of $G$, is an upper-bound on the number of dominating vertices used to fully dominate $G$ by diagonalization and projection. Since the size of super-grid of and $m\times n$ grid $G$ is $(m+2)\times (n+2)$, therefore, $|U|'\le \left\lceil \frac{(m+2)(n+2)}{5}\right\rceil$. Hence, $\left\lceil \frac{(m+2)(n+2)}{5}\right\rceil$ is an upper-bound on the number of dominating vertices used to dominate grid $G$ by Chang's algorithm. The following theorem reflects this upper-bound.

\begin{theorem}[Chang $\etal$,~\cite{TYC'92}]
\label{theo:upper-bound}
For any $m\times n$ grid $G=(V,E)$ with $m,n\in \mathbb{N}$, a dominating set $S\subset V$ can be constructed in polynomial-time, such that $|S|\le \left\lceil \frac{(m+2)(n+2)}{5}\right\rceil$. Moreover, for grids with $16\le m\le n$ we have $|S|-\gamma_{m,n}\le 5$.
\end{theorem}

The upper-bound on the difference between the cardinality of the provided dominating set $S$ from the domination number of an $m\times n$ grid $G$ with $16\le m\le n$, $\gamma_{m,n}$, is obtained by virtue of Theorem \ref{theo:cited-4}. An example of constructing dominating sets for grids using diagonalization and projection is shown in Figure \ref{fig:eg}.

In the following lemma we show that although in diagonal patterns no vertex is covered by more than one dominating vertex, using a simple greedy algorithm does not necessarily result in diagonalizing the grid or using at most $\left\lceil \frac{(m+2)(n+2)}{5}\right\rceil$ dominating vertices to dominate the grid. 

\begin{lemma}
\label{lem:greedy-grid}
The size of the dominating set obtained by a greedy algorithm on an $m\times n$ grid $G$ might be as large as $\left\lceil \frac{m}{3}\right\rceil \left\lceil \frac{n}{3}\right\rceil+2\left\lfloor \frac{m}{3}\right\rfloor \left\lfloor \frac{n}{3}\right\rfloor$.
\end{lemma}

\begin{proof}
As discussed in Section~\ref{sec:intro}, after the first vertex $v$ is added to the dominating set $S$, greedy algorithm chooses a vertex that does not share any neighbours with $v$. Although this is also a property of diagonal patterns, the set of all the closest vertices around $v$ that can be added to $S$ using diagonal patterns has size at most four (see Figure~\ref{fig:patterns}(b)). However, there are 12 vertices around $v$ that do not share any neighbours with $v$ and hence candidate to be added to $S$ in a  greedy algorithm, Figure~\ref{fig:greedy}(a). At each step of a greedy algorithm one of these 12 vertices is chosen arbitrarily. However, choosing only all red vertices or all blue vertices would start developing a diagonal pattern. Other combinations of candidate vertices would fail to diagonalize the grid and some vertices of the graph would be dominated by more than one dominating vertex. Hence, the size of the constructed dominating set would be greater than $\left\lceil \frac{(m+2)(n+2)}{5}\right\rceil$.

In particular, the algorithm might add all the green vertices to $S$ and repeat the same pattern in the grid, Figure~\ref{fig:greedy}(b). However, using this pattern, between any four green vertices there remains a set of four vertices that are not dominated by any vertex in $S$. These vertices are highlighted by dotted rectangles in Figure~\ref{fig:greedy}(b). To dominate each of these sets of vertices at least two extra dominating vertices should be added to $S$. Therefore, the number of obtained dominating vertices would be at least $\left\lceil \frac{m}{3}\right\rceil \left\lceil \frac{n}{3}\right\rceil+2\left\lfloor \frac{m}{3}\right\rfloor \left\lfloor \frac{n}{3}\right\rfloor$, which is much greater than the size of the dominating set obtained by Chang's construction, i.e., $\left\lceil \frac{(m+2)(n+2)}{5}\right\rceil$.
\end{proof}

\begin{figure}[t]
\includegraphics[scale=0.3]{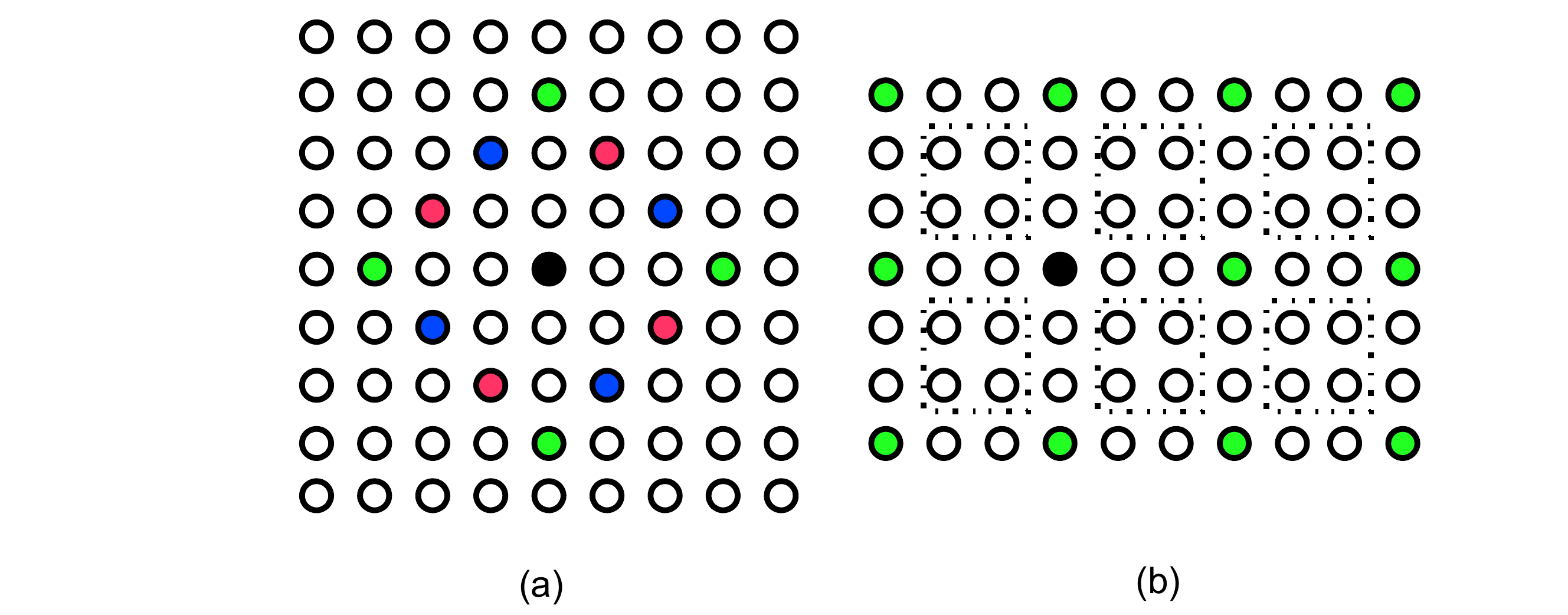}
\caption{In Figure (a), after adding the black vertex to the dominating set, the next vertex added to that set can be any of the blue, red or green vertices, without dominating any vertex by two dominating vertices. In (b), a dominating set is built by starting from the black vertex and keep adding the green vertices shown in (a) to the set. Each dotted rectangle contains four vertices that are not dominating by the obtained dominating set.}
\label{fig:greedy}
\end{figure}

\section{Distributed Grid Domination}
\label{sec:dist-domin}
In the preceding section, a centralized algorithm was discussed that produced a dominating set $S$ for a given $m\times n$ grid $G$ such that $|S|\le \left\lceil \frac{(m+2)(n+2)}{5}\right\rceil$. In this section, we show how to achieve the same upper-bound in a distributed way.

\subsection{Model and Notation}
\label{subsec:model2}
Here we assume that the environment is an $m\times n$ grid $G=(V,E)$ with $m,n\in \mathbb{N}$. The goal is to dominate the grid environment in a distributed fashion using several robots (or agents) without any knowledge of environment size. Initially, there exist $k$ agents in the environment, where $k$ can be smaller or greater than the number of agents needed to dominate the grid. The following assumptions are made for the grid and agents.

\emph{Grid Assumptions:} Agents can be located only on the vertices of the grid and are able to move between the grid vertices only on the edges of the grid. At each moment, a vertex can contain more than one agent. We refer to the vertices using the standard Cartesian coordinates defined in Section \ref{sec:background}.

\emph{Agent Assumptions:} The agents, denoted by $a_1,\ldots, a_{k}$, are initially located at arbitrary vertices on the grid.
The agents have three modes: (a) \emph{sleep}, (b) \emph{active}, and (c) \emph{settled}. The mode of an agent $a$ and the vertex it is located at are denoted by $\mathrm{mode}(a)$ and $v(a)$, respectively. Only agents in the active and settled modes are able to communicate. At the beginning of the procedure, all the agents are in the sleep mode. During each \emph{epoch}, that is, a time interval with a specified length, one agent goes to active mode. The activation sequence of agents is arbitrary (e.g., it can be scheduled in advance or it can be random). The active agent can communicate with the settled agents to perform the distributed dominating set algorithm. Once an agent activates and performs its part in the algorithm, it goes to settled mode. Ultimately, all the settled agents go back to sleep mode and will not activate again.

Here, each agent is equipped with suitable angle-of-arrival (bearing) and range sensors. Using these sensors, agent $a$ computes the coordinates of other agents in its own coordinate frame $\Sigma_a$ with its origin at $v(a)$ and an arbitrary orientation, fixed relative to agent $a$. Each agent also has a compass to determine its heading direction. Additionally, agents are equipped with short-ranged proximity sensors to sense the environment boundary. Agents are able to sense the boundary only if they are on a vertex $v$ whose neighbour is a boundary vertex of the grid, i.e., $N(v)\cap B(G)\ne \emptyset$. The compass helps agents to distinguish which of the four boundary edges they are approaching.

\subsection{Overview of Algorithm}
\label{subsec:dist-alg-summary}
The main idea in this algorithm is to implement the diagonal pattern defined in Section \ref{sec:center-domin} on grid $G=(V,E)$, using communications among active and settled agents. A special unit called a \emph{module} is defined for the active and settled agents. A module is a cross-like shape consisting of the agent at its center with the associated dominated vertices in the arms of the cross (see Figure \ref{fig:patterns}(b)).
For each module $m$, the vertex that contains the agent, i.e., the center vertex, is referred to as the \emph{module center}, denoted by $c(m)$. As an agent moves on the grid to contribute to the diagonal pattern, its module moves with it as well. Modules $m_1$ and $m_2$ with module centers $c(m_1)=v_{i,j}$ and $c(m_2)=v_{i',j'}$ can connect to each other if $v_{i',j'}\in \{ v_{i+1,j+2},v_{i+2,j-1},v_{i-1,j-2},v_{i-2,j+1}\}$ (see Figure \ref{fig:steps}(f)). This condition is called the \emph{module connection condition}. The set of centers of the connected modules is called a \emph{cluster}. We will later show that the module connection condition ensures that the module centers are a diagonalization of the vertices covered by the modules in the cluster.

\emph{Valid Slots:} Let $G'=(V',E')$ be the super-grid of $G$. A vertex $v_{a,b}\in V'$ is called a \emph{slot} if there exists a module $m$ in the cluster with center $v_{i,j}$ such that $v_{a,b}\in \{ v_{i+1,j+2},v_{i+2,j-1},v_{i-1,j-2},v_{i-2,j+1}\}$ and $v_{a,b}$ is not already a center for a module in the cluster.
For a settled agent $a$ located at $v(a)$, denote the set of all its slots by $\mathrm{slots}(a)$. Recall that the \emph{orphan} of a vertex $v\in V'\backslash V$, i.e., $\mathrm{orphan}(v)$, is a vertex $u\in B(G)$ such that $u\in N(v)$. The set of all \emph{valid slots} for settled agent $a$, denoted by $\mathrm{vslots}(a)$, is defined as $(\mathrm{slots}(a)\cap V) \cup \mathrm{orphan}(\mathrm{slots}(a)\backslash V)$. Newly activated agents can settle only on the valid slots of the settled agents.

\emph{Updating Valid Slots:} When an active agent settles, it creates the list of its valid slots as follows. If a settled agent $a$ cannot sense the boundary (i.e., it has no neighbour on the boundary), $\mathrm{slots}(a)\backslash V=\emptyset$ and hence $\mathrm{vslots}(a)=\mathrm{slots}(a)$. Conversely, a settled agent can also determine which of its slots lie outside the grid boundary (Figure \ref{fig:orphans}(a)). Each newly settled agent marks the vertices on the grid boundary that are neighbours of $\mathrm{slots}(a)\backslash V$ as \emph{orphans} and so $\mathrm{vslots}(a)=(\mathrm{slots}(a)\cap V) \cup \mathrm{orphan}(\mathrm{slots}(a)\backslash V)$ (Figure \ref{fig:orphans}(b)). By the definition of valid slots, no valid slot exists in an orphan's neighbourhood. Therefore, each orphan needs one agent to be located on itself or one of its neighbours to be dominated. For simplicity we always put an agent on the orphan itself.

When an agent activates, it transmits a signal to find the settled agents on the grid and waits for some specified time for a response from them. Since there is no settled agent in the environment when the first agent activates, it receives no signal and concludes it is the  first one activated. Thus, the agent stays at its initial location and goes to the settled mode. Subsequently, each active agent translates to the closest settled agent.\footnote{Note that for completeness of the algorithm, it is not necessary for the active agents to go to the closest settled agents. An active agent can go toward any arbitrary settled agent to occupy its valid slot.}

\subsection{Distributed Grid Domination Algorithm}
\label{susec:dist-alg}
During the distributed grid domination algorithm, active agents can either contribute to grid diagonalization by locating on non-orphan valid slots or can settle on orphans. In each epoch, the set of the non-orphan vertices containing the previously settled agents is called the \emph{cluster} and is denoted by $C$, while the set of occupied orphans is denoted by $P$. At the beginning of the algorithm $C=P=\emptyset$. It should be mentioned that $C$ and $P$ are not saved by any agent, and are used only to aid in the presentation of the algorithm. Moreover, we denote the set of all settled agents at each moment by $A_s$, where at the beginning of the algorithm $A_s=\emptyset$. Also if agent $a$ is already settled and is now in sleep mode $\mathrm{done}(a)=1$, otherwise $\mathrm{done}(a)=0$.

\begin{algorithm2e} 
  \small 
  \DontPrintSemicolon 
  \KwIn{An $m\times n$ Grid and a set of agents $A$}
  \While {$\exists$ agent $a\in A$ with $\mathrm{mode}(a)=sleep$ and $\mathrm{done}(a)=0$}{
  $\mathrm{mode}(a):=\mathrm{active}$, $a$ sends out signal to $A_s$ (Figure \ref{fig:steps}(b)). \;
  \If{$A_s\neq \emptyset$}{
   At least one agent in $A_s$ sends a signal out to $a$. \; 
   }
  \If{$a$ receives no signal}{
    $\mathrm{mode}(a):=\mathrm{settled}$ (Figure \ref{fig:steps}(a)).
    \;
        $A_s:=\{a\}$. \;
        $C:=\{v(a)\}$. \;
        Skip to Line 22. \;
      }
      \eIf{$\mathrm{vslots}(A_s)\neq\emptyset$}{
      Agent $a$ computes the closest settled agent $s\in A_s$ and
      notifies $A_s$. \;
      Agent $s$ sends the coordinates of $\mathrm{vslots}(s)$ to $a$. \;
      Agent $a$ moves toward the closest $v\in\mathrm{vslots}(s)$.\;
    \If{$v(a)=v$}{
        $\mathrm{mode}(a):=\mathrm{settled}$ (Figure \ref{fig:steps}(d)). \;
        $A_s:=A_s \cup \{a\}$. \;
    }
    \eIf{$v(a)$ and $v(s)$ satisfy the module connection condition}{
        $C:=C\cup \{v(a)\}$. \;
    }{
       $P:=P\cup \{v(a)\}$ (Figure \ref{fig:orphans}(c)). \;
        $\mathrm{mode}(a):=\mathrm{sleep}$. \;
    }
    \For{$i = 1 \to |A_s|$}{
        \If{$v(A_s(i))\in C$ and $\mathrm{mode}(A_s(i))\neq \mathrm{sleep}$}{
            Update $\mathrm{vslots}(A_s(i))$ (Figures \ref{fig:steps}(e) and \ref{fig:orphans}(d)). \;
            \If{$\mathrm{vslots}(A_s(i))=\emptyset$}{
                $\mathrm{mode}(A_s(i)):=\mathrm{sleep}$ (Figure \ref{fig:steps}(f)). \;
                $\mathrm{done}(A_s(i)):=1$.\;
            }
        }
    }
    }{
    Break.\;
    }
    }
    The remaining non-activated agents leave the grid. \;
    \caption{\textsc{Distributed Grid Domination}}
    \label{alg:dist-domin}
\end{algorithm2e}

\begin{remark}(Comments on Algorithm)\\
1) Since agents can move only on the grid edges, the distance
  between two vertices can be computed simply by adding their
  $x$-coordinate and $y$-coordinate differences, i.e., $\Delta x$ and
  $\Delta y$. There exist many shortest paths between any two vertices
  and agent $a$ arbitrarily chooses one of them to traverse; for
  instance it can first traverse on the $x$-coordinate and then on the
  $y$-coordinate. \\
2) In Step 11, agent $a$ locates $s$ in $\Sigma_a$ (i.e., coordinate frame of $a$), while $\mathrm{vslots}(s)$ is computed by $s$ in
  $\Sigma_s$ in Step 12. In Step 13, agent $a$ converts the coordinates
  of $\mathrm{vslots}(s)$ from $\Sigma_s$ to $\Sigma_a$ for
  traversing, using relative sensing techniques
  \cite{PSFBA'08}. \\
3) When an agent settles, all settled agents wait for a specified amount of time for the next agent to activate. If no agent activates, Algorithm \ref{alg:dist-domin} halts and the previously settled agents construct a subset of a dominating set of the grid. This happens when the initial number of agents is not sufficient to dominate the grid.\\
4) If the agents are equipped with GPS, then they can agree on a fixed diagonalization (i.e., agree on a value of $r$), and move to the vertices $U$ in the diagonalization.  At this point, only orphan vertices exist. The remaining agents can move along the boundary to find and cover all orphans and consequently dominate the grid. Hence, in this paper we study the case that agents are not armed with GPS.
\end{remark}

\begin{figure}[t]
\includegraphics[scale=0.27]{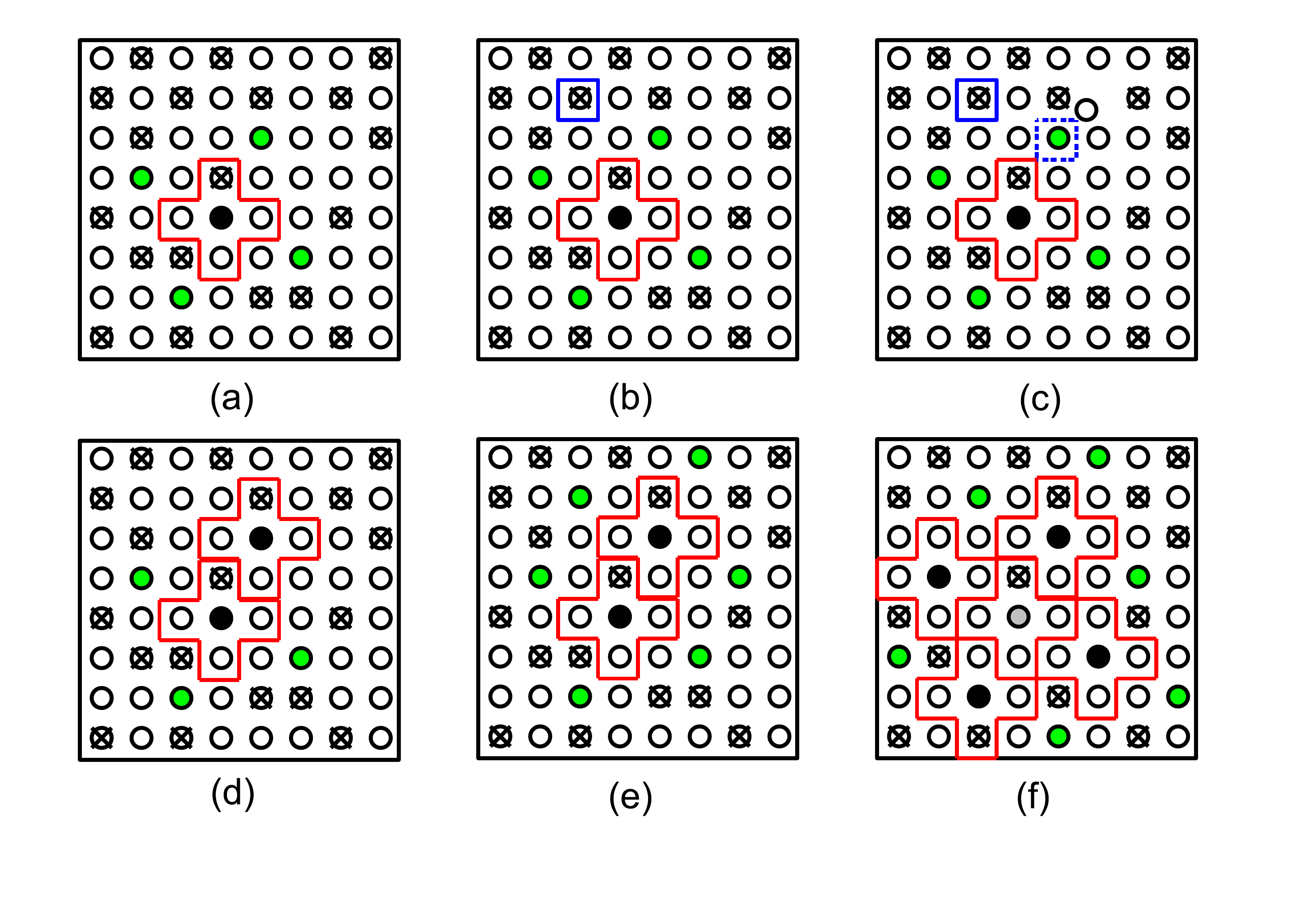}
\vspace{-1cm}
\caption{Non-activated agents are marked by black crosses and the already settled agents are shown by black circles. Agents in $C$ have red crosses as their modules. Figure (a) shows the first active agent, as in Step 6. In (b), an active agent is highlighted by a blue square. Step 13 is depicted in (c), where a dashed blue square shows the closest valid slot to the active agent. In (d), the active agent moves to the valid slot and joins $C$, as in Step 15. In (e), the list of valid slots is updated as in Step 4. In (f), the grey circle shows an agent that goes from settled to sleep mode.}
\label{fig:steps}
\end{figure}

\begin{figure}[t]
\includegraphics[scale=0.27]{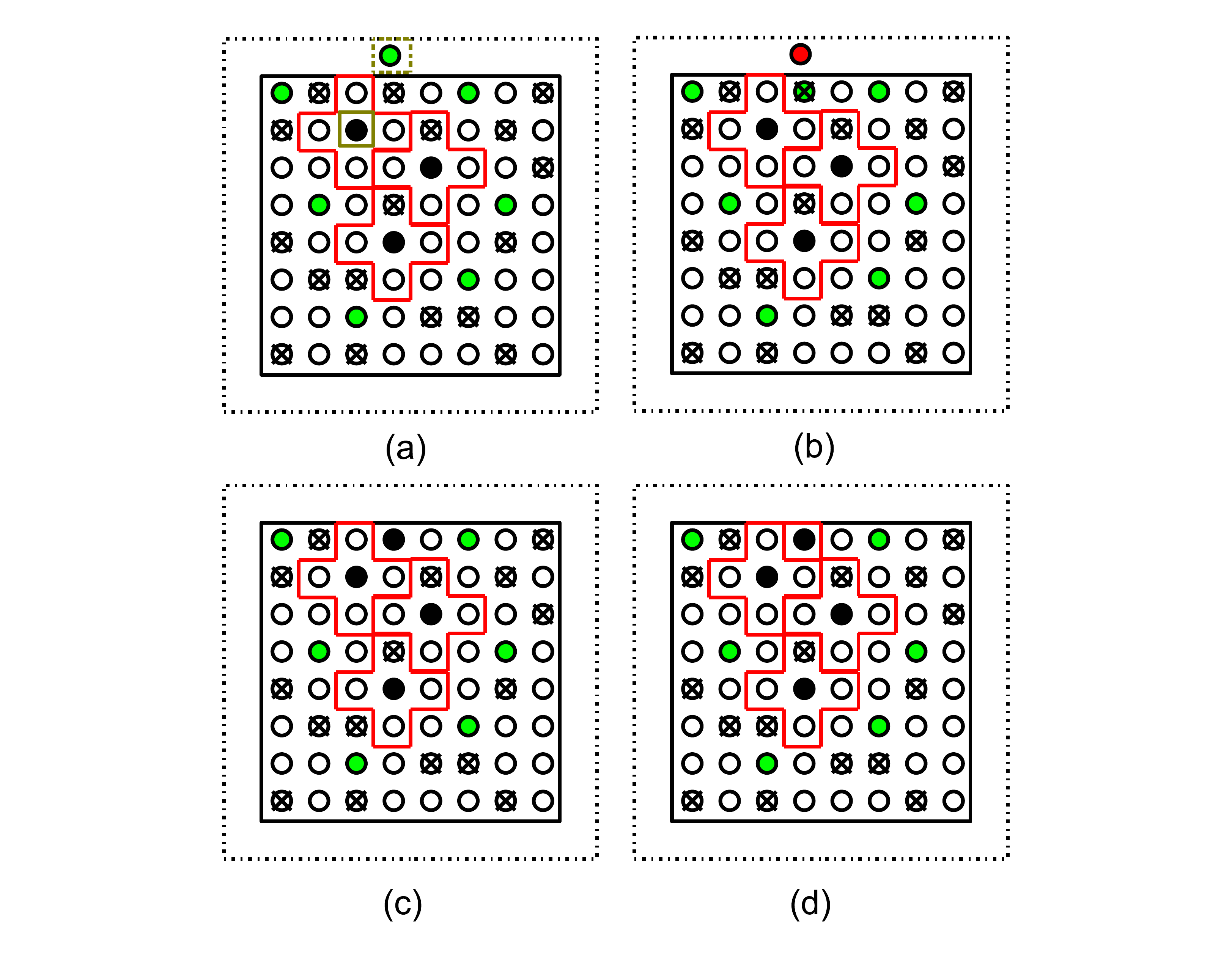}
\vspace{0cm}
\caption{Non-activated agents are marked by black crosses and the already settled agents are shown by black circles, with red crosses as their modules. In (a), a settled agent, highlighted by a solid blue square, realizes one of its slots, shown by a dashed blue square, is outside the grid boundary. In (b), the settled agent replaces the slot outside the grid boundary with its orphan and name the resulting set as its valid slots. In (c), the active agent locates at the orphan. Figure (d) shows that an agent on an orphan has no valid slot.}
\label{fig:orphans}
\end{figure}

\subsection{Distributed Algorithm Analysis}
\label{subsec:dist-alg-analys}
We now prove that the set of vertices determined by Algorithm \ref{alg:dist-domin}, i.e., $C\cup P$, creates a dominating set for the grid. Recall that at each epoch, $C$ is the set of non-orphan vertices containing the previously settled agents and $P$ is the set of occupied orphans.

\begin{lemma}
\label{lem:modules-are-diag}
During the operation of Algorithm \ref{alg:dist-domin}, the module connection condition forces the vertices in $C$ to create a diagonal pattern.
\end{lemma}

\begin{proof}
This will be proved using induction on the size of $C$ during the operation of the algorithm. According to the module connection condition, the module of agent $a$ located at vertex $v(a)=v_{i',j'}\notin C$ can connect to the module of vertex $v_{i,j}\in C$ if $v_{i',j'}\in \{v_{i+1,j+2},v_{i+2,j-1},v_{i-1,j-2},v_{i-2,j+1}\}$. The base of induction is $|C|=0$, when the first agent is about to be added to $C$. In this case, the first agent settles at its current location $v(a)=v_{i,j}$ and establishes the value $r\equiv j-2i \pmod{5}$.

For $|C|>1$, $C$ already has a diagonal pattern and an active agent $a$ at $v(a)=v_{i',j'}$ aims to join it by connecting to a module centered at $v_{i,j}$. Since $v_{i,j}$ is already in $C$, $j-2i\equiv r \pmod{5}$. It can be seen that for a vertex $v_{i',j'}$ that satisfies the module connection condition with respect to $v_{i,j}$ we have $j'-2i'\equiv r \pmod{5}$. Therefore, the resulting set has a diagonal pattern.
\end{proof}

\begin{theorem}
\label{theo:error-dist}
The number of agents used to dominate an $m\times n$ grid $G=(V,E)$ by Algorithm \ref{alg:dist-domin} is upper-bounded by $\left\lceil \frac{(m+2)(n+2)}{5}\right\rceil$. For grids with $16\le m\le n$, the number of agents used is upper-bounded by $\gamma_{m,n}+5$.
\end{theorem}

\begin{proof}
We first prove Algorithm \ref{alg:dist-domin} is correct and then show the upper-bound holds. Let $G'=(V',E')$ be the super-grid of $G$ and
$C$ denote the non-orphan vertices occupied by previously settled agents when the algorithm finishes. By Lemma \ref{lem:modules-are-diag}, $C$ constitutes a diagonal pattern and by the condition in Step 10 of the algorithm no other agent can be added to $C$; therefore, $C$ diagonalizes $G$. Moreover, orphans are neighbours of the vertices in $V'\backslash V$ that are initially detected as slots by the settled agents and hence diagonalize $G'$ by Lemma \ref{lem:modules-are-diag}. Thus, locating one agent on each orphan is equivalent to the projection process. Hence, if a sufficient number of agents exist in the grid, Algorithm \ref{alg:dist-domin} provides a dominating set for $G$ (from Theorem \ref{theo:upper-bound}). Consequently, the algorithm is \emph{complete}, meaning it always finds a solution, if one exists.

Furthermore, since Algorithm \ref{alg:dist-domin} performs diagonalization and projection on $G$, from Theorem \ref{theo:upper-bound} it immediately follows that the number of agents used in the algorithm, $n_a$, is upper-bounded by $\left\lceil \frac{(m+2)(n+2)}{5}\right\rceil$. Also by Theorem \ref{theo:cited-4}, for $16\le m\le n$ we have $n_a-\gamma_{m,n}\le 5$.
\end{proof}

Note that while the agents do not form a dominating set for $G$, an active agent finds a valid slot in at most $n+m$ steps. A step is a specified time duration within which an agent performs its basic operation, such as traversing an edge or transmitting signals. Since the number of agents needed to dominate an $m\times n$ grid is less than $mn$, Algorithm \ref{alg:dist-domin} takes at most $mn(m+n)$ steps to construct a dominating set for $G$.

\subsection{Simulations}
\label{subsec:simulation}
To augment and examine the results discussed in this section, we simulated Algorithm \ref{alg:dist-domin} on various grids and different initial configurations of agents on grid vertices. Figure \ref{fig:sim1}(a) demonstrates a $10\times 15$ grid graph with 41 agents located randomly on it. The first agent that activates is located on vertex $(5,9)$ and hence stays on that vertex. Figure \ref{fig:sim1}(b) shows the location of agents when Algorithm \ref{alg:dist-domin} is complete. It can be seen that every vertex is dominated. However, there are some agents located at vertices (such as $(6,5), $$(6,12)$ and $(7,15)$) that are never activated in the algorithm. These are the additional agents that are not required to dominate the grid and they are removed in Figure \ref{fig:sim1}(c).

\begin{figure}[t]
        \begin{subfigure}[b]{0.28\linewidth}
                \centering
                \includegraphics[width=\linewidth]{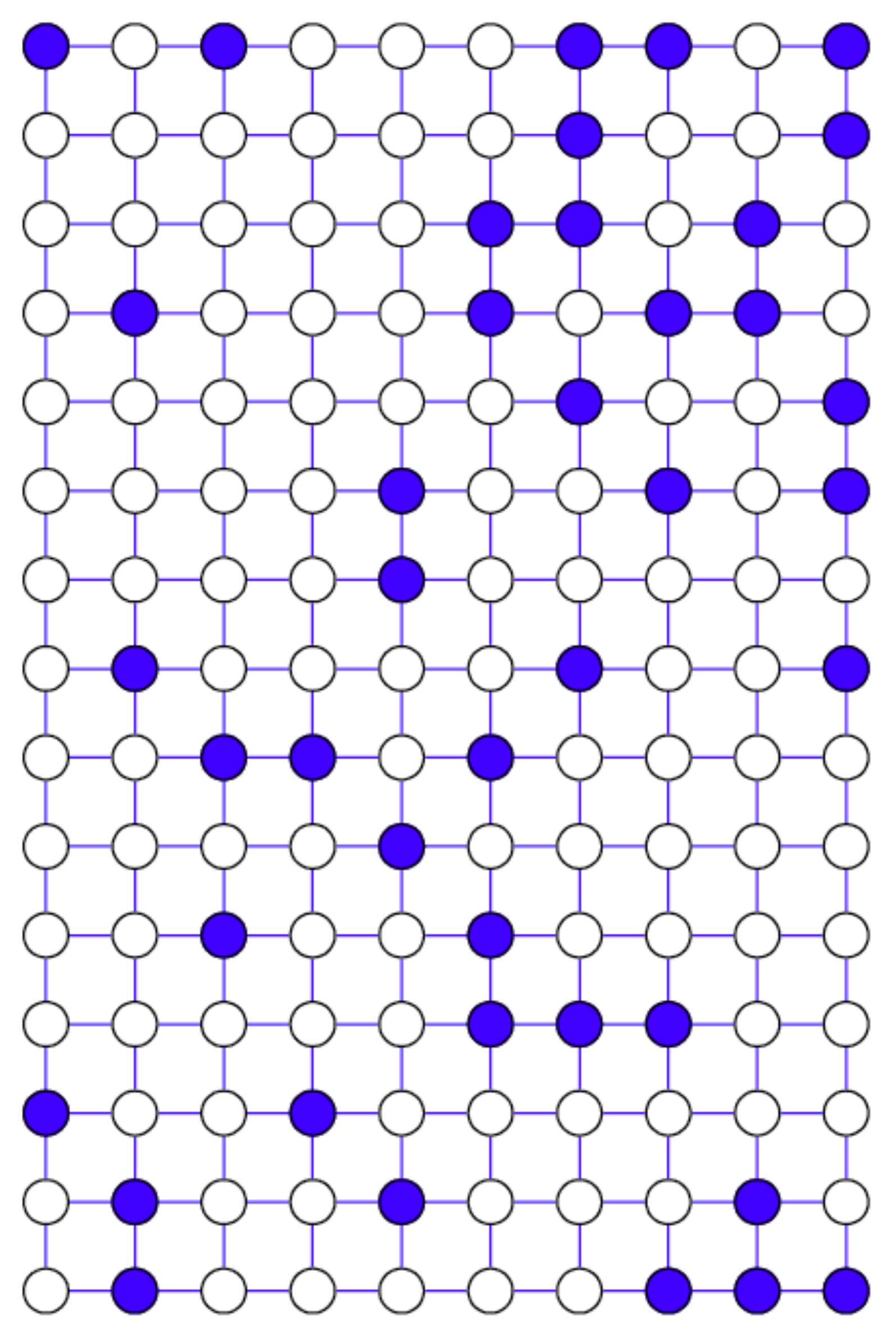}
                \caption{}
        \end{subfigure}
        ~ 
  \hfill
        \begin{subfigure}[b]{0.28\linewidth}
                \centering
                \includegraphics[width=\linewidth]{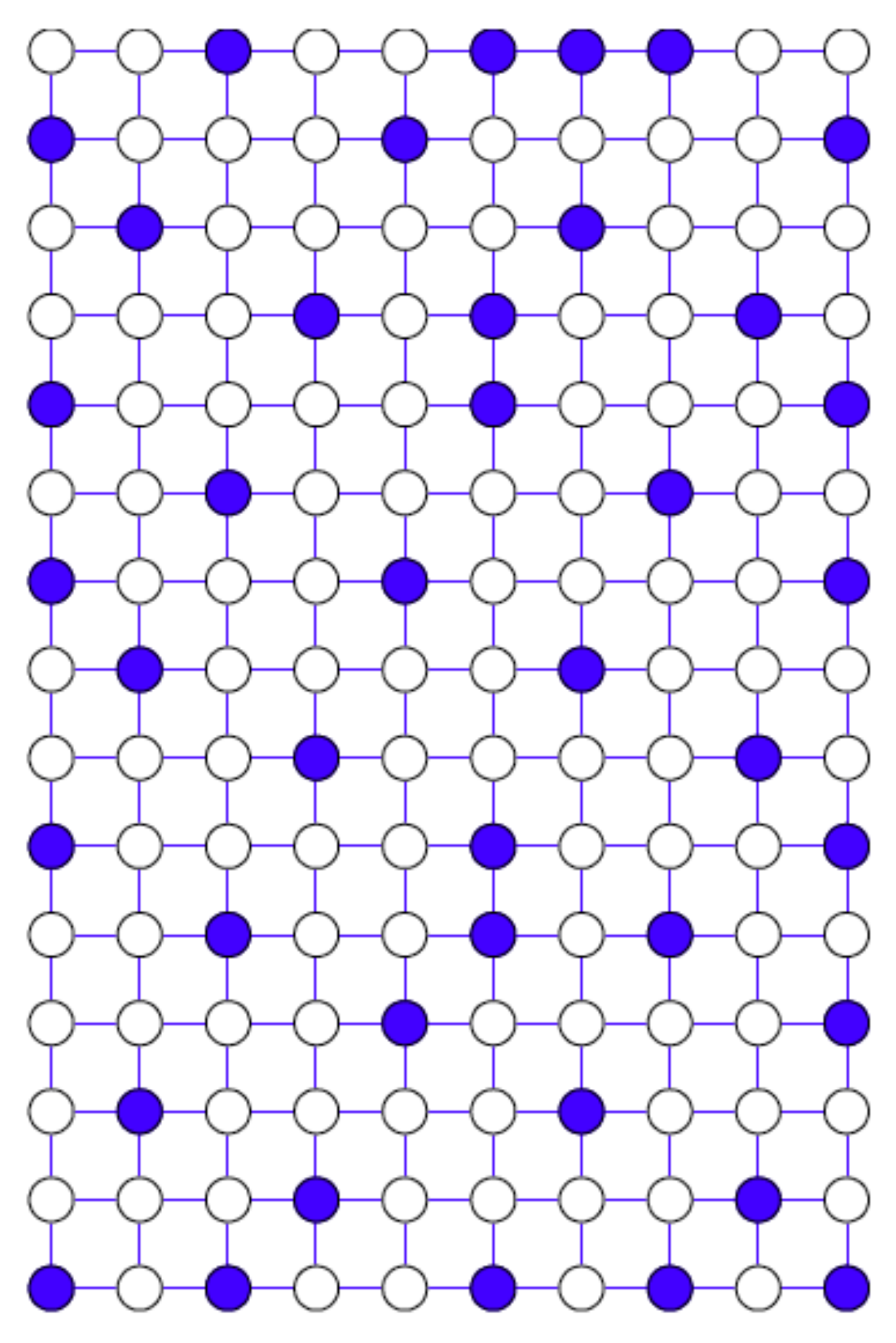}
                \caption{}
        \end{subfigure}
        ~
  \hfill
 \begin{subfigure}[b]{0.28\linewidth}
                \centering
                \includegraphics[width=\linewidth]{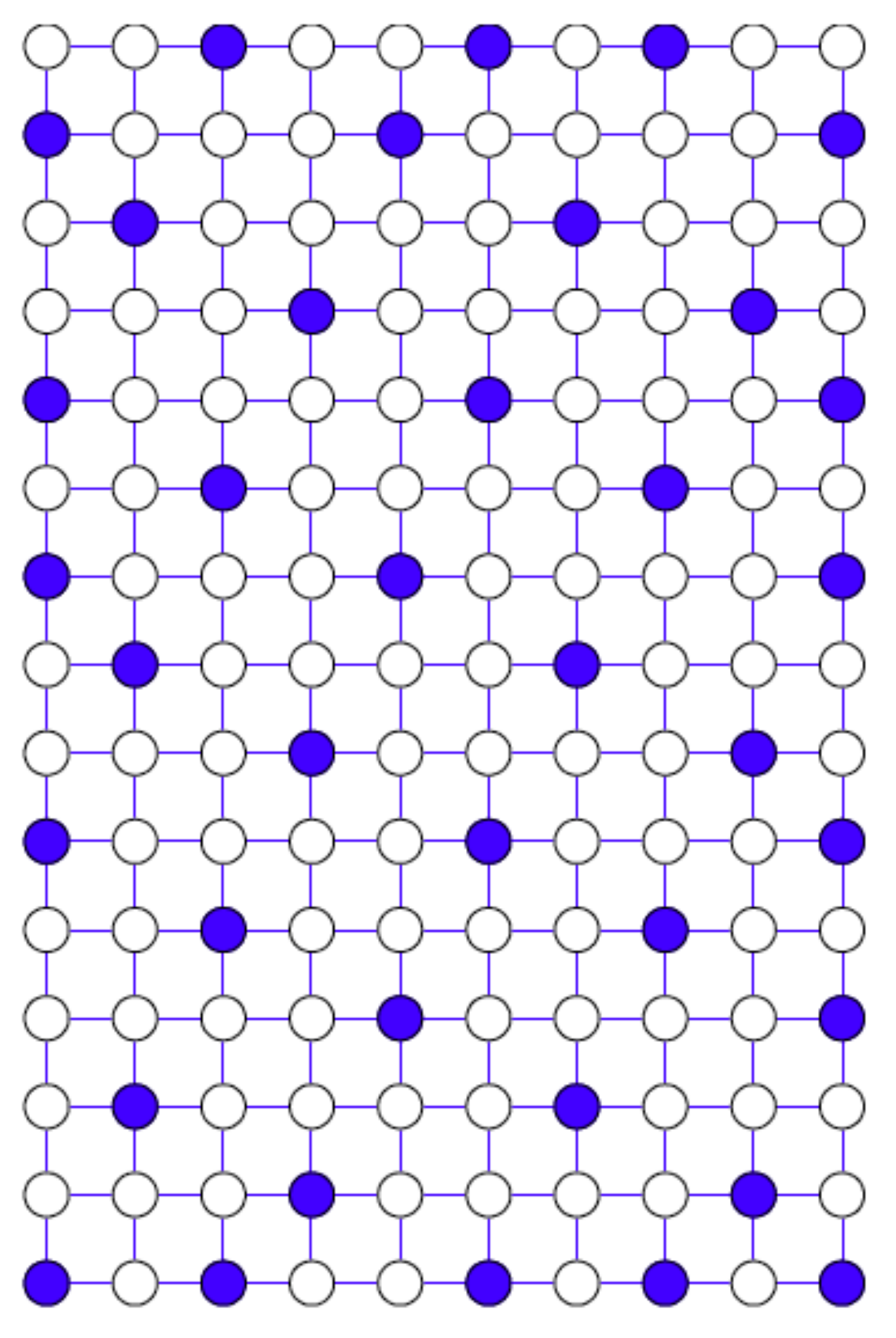}
                \caption{}
        \end{subfigure}
        ~
        \caption{A $10\times 15$ grid is depicted with agents shown in blue. In (a) the initial configuration of the agents is shown and (b) shows the agents configuration when Algorithm \ref{alg:dist-domin} is finished. In (c), all non-settled and non-asleep agents leave the grid}
        \label{fig:sim1}
\end{figure}

\section{$k$-Distance Domination on Grids}
\label{sec:k-distance}
In this section we generalize Chang's algorithm for grid domination, discussed in Section~\ref{sec:center-domin}, to the $k$-distance dominating set problem, where a vertex dominates all the vertices within distance $k$ from it. Before defining the problem formally, let $d(u,v)$ denote the shortest path distance between vertices $v,u\in V$ in $G=(V,E)$. Moreover, vertex $u\in V$ is defined as a \emph{$k$-neighbour} of vertex $v\in V$, if $0<d(u,v)\le k$. The set of all $k$-neighbours of $v$ is denoted by $N^k(v)$. Moreover, for a set of vertices $W\subset V$ and a vertex $v\in V\backslash W$, we have $u=\mathrm{friend}^k(v,W)$ if (a) $u\in W$, (b) $u\in N^k(v)$, and (c) $d(v,u)\le d(v,w),\forall w\in W$.

\begin{definition}($k$-Distance Dominating Set Problem)
\label{def:k-distance}
Given a graph $G=(V,E)$, the $k$-distance dominating set problem is to find a set of vertices $S\subseteq V$ such that for every vertex $v\in V\backslash S$ there exists a vertex $u\in S$ where $u\in N^k(v)$. The cardinality of a smallest $k$-distance dominating set for $G$ is called the \emph{$k$-distance domination number} of $G$ and is denoted by $\gamma^k(G)$ \cite{Henning'98}.
\end{definition}

We say that vertex $u\in S$ \emph{$k$-distance dominates} $v\in V\backslash S$ if $d(u,v)\le k$. The regular dominating set problem is a special case of the $k$-distance dominating set problem, where $k=1$. Therefore, $k$-distance domination is also NP-hard on general graphs. However, to the best of our knowledge the $k$-distance domination number of grids is not known and the complexity of the problem is open. In Section \ref{subsec:upper-bound-k-distance}, we
generalize the approaches in Sections \ref{sec:center-domin} and \ref{sec:dist-domin} to provide a $k$-distance dominating set for an $m\times n$ grid graph $G$.

\subsection{Centralized $k$-Distance Domination on Grids}
\label{subsec:upper-bound-k-distance}
Before discussing the $k$-distance domination algorithms on grids we introduce the following definitions.

\begin{definition}($k$-Sub-Grids and $k$-Super-Grids)
\label{def:k-super}
An $m\times n$ grid $G=(V,E)$ is called a \emph{$k$-sub-grid} of an $m'\times n'$ grid $G'=(V',E')$ if $G$ is induced by vertices $v'_{i,j}\in V'$, where $k+1\le i\le m'-k$ and $k+1\le j\le n'-k$. If $G$ is a $k$-sub-grid of $G'$, $G'$ is called the \emph{$k$-super-grid} of $G$.
\end{definition}

\begin{lemma}
\label{lem:k-distance-unit}
For an $m\times n$ grid $G=(V,E)$, $|N^k(v)|\le 2k^2+2k+1$.
\end{lemma}

\begin{proof}
Since $G$ is a grid, the $k$-neighbours of $v$ form a diamond around it with a diameter of $2k+1$ (see the red regions in Figure \ref{fig:k-distance-elem}). Thus $|N^k(v)|$ is upper-bounded by the area of this region, which is $\left\lceil \frac{(2k+1)^2}{2}\right\rceil=2k^2+2k+1$.
\end{proof}

In what follows we define $N^k_{\max}=2k^2+2k+1$.

\begin{definition}($k$-Diagonal Pattern)
\label{def:k-diag-patt}
A set of vertices $U\subset V$ constitutes a \emph{$k$-diagonal pattern} on grid $G=(V,E)$ if there exists a fixed $0\le r< N^k_{\max}, r\in \mathbb{Z}_+$ such that for any vertex $v_{x,y}\in U$ we have $ky-(k+1)x\equiv r \pmod{N^k_{\max}}$ (see Figure~\ref{fig:k-distance-elem}).
\end{definition}

\begin{definition}($k$-Diagonalization)
\label{def:k-diagonalization}
A set of vertices $U\subset V$ \emph{$k$-diagonalizes} grid $G=(V,E)$ if it constitutes a $k$-diagonal pattern and there exists no vertex $v\in V\backslash U$ that can be added to $U$ so that $U$ remains a $k$-diagonal pattern.
\end{definition}

Moreover, for a grid $G=(V,E)$ and its $k$-super-grid $G'=(V',E')$, the \emph{$k$-projection} is defined as a special mapping from the vertices in $V'\backslash V$ to their $k$-neighbours in $V$. It is defined formally as follows.

\begin{definition}($k$-Projection)
\label{def:k-porjection}
Consider a grid $G=(V,E)$ and its $k$-super-grid $G'=(V',E')$. The \emph{$k$-projection} for a set $U'\subseteq V'$ is defined as
the set $U''=\{u\in V|~\exists v\in U'\backslash V~s.t.~u=\mathrm{friend}^k(v,V) \}\cup \{U'\cap V\}$ (see Figure \ref{fig:k-grid}).
\end{definition}

\begin{lemma}
\label{lem:k-min-2k+1}
Let $U$ be a set of vertices that $k$-diagonalizes a grid $G=(V,E)$. For any two vertices $v_{x,y},v_{x',y'}\in U$ we have $d(v_{x,y},v_{x',y'})\ge 2k+1$.
\end{lemma}

\begin{proof}
Since $v_{x,y},v_{x',y'}\in U$, we have $y=\frac{1}{k}((k+1)x+r+qN^k_{\max})$ and $y'=\frac{1}{k}((k+1)x'+r+q'N^k_{\max})$, where $r,q\in \mathbb{Z}$ and $0\le r< N^k_{\max}$. We define $\Delta_q=q'-q$, $\Delta_1=x'-x$ and $\Delta_2=y'-y=\frac{k+1}{k}\Delta_1+\frac{N^k_{\max}}{k}\Delta_q$. The shortest distance between $v_{x,y},v_{x',y'}$ is equal to $|\Delta_1|+|\Delta_2|$. From $\Delta_2=\frac{k+1}{k}\Delta_1+\frac{N^k_{\max}}{k}\Delta_q$ it can be observed that as $\Delta_1$ grows, $\Delta_2$ grows faster compared to $\Delta_1$. Hence $|\Delta_1|+|\Delta_2|$ is minimum when $\Delta_2=0$ and $|\Delta_1|=|\frac{N^k_{\max}}{k+1}\Delta_q|$. Note that the minimum (non-zero) distance occurs for $\Delta_q=1$ and also it is an integer, hence it is lower-bounded by $\left\lceil{\frac{2k^2+2k+1}{k+1}}\right\rceil=2k+1$.
\end{proof}

\begin{lemma}
\label{lem:k-distane-k-super}
Consider a grid $G=(V,E)$ and its $k$-super-grid $G'=(V',E')$. If $U'\subset V'$ $k$-diagonalizes $G'$, then each vertex in $V$ is $k$-dominated by exactly one vertex from $U'$.
\end{lemma}

\begin{proof}
For each vertex $v_{x,y}\in V$ let $r_{v_{x,y}}\equiv ky-(k+1)x\pmod{N^k_{\max}}$. Consider any vertex $v\in V$ and its $k$-neighbourhood $N^k(v)$. The distance between any two vertices in $J=\{v\}\cup N^k(v)$ is at most $2k$. Also, there are exactly $N^k_{\max}$ vertices in this set. Thus, for any two distinct vertices $u,w\in J$ we have $r_u\neq r_w$ by Lemma \ref{lem:k-min-2k+1}. Hence each vertex $u\in N^k(v)$ has a distinct value of $r_u$. Consequently, for the value of $r$ that corresponds to the diagonalization $U'$, there is exactly one vertex in the $k$-neighbourhood of $v$ such that $r_v=r$ and thus $v$ is $k$-dominated by exactly one vertex from $U'$.
\end{proof}

\begin{lemma}
\label{lem:k-diag-cardinal}
If a set of vertices $U\subset V$ $k$-diagonalizes an $m\times n$ grid $G=(V,E)$, then $U$ contains at most $\left\lceil \frac{mn}{ N^k_{\max}}+\frac{N^k_{\max}}{4}\right\rceil$
vertices.
\end{lemma}

\begin{proof}
Since $U$ $k$-diagonalizes $G$, it constitutes a $k$-diagonal pattern on $G$ such that no more vertices can be added to it while maintaining a $k$-diagonal pattern. Therefore, among each $N^k_{\max}$ consecutive vertices in any row or column there is exactly one vertex from $U$. Hence, the number of vertices of $U$ in a row/column of $t$ vertices is at most $\left\lceil\frac{t}{N^k_{\max}}\right\rceil$.

Thus, there are at most $N^k_{\max}$ vertices from $U$ in any $N^k_{\max}\times N^k_{\max}$ grid. Hence, in any $N^k_{\max}q\times N^k_{\max}p$ grid with $p,q\in \mathbb{Z_+}$, there are at most $N^k_{\max}pq$ vertices from $U$. For an $m\times n$ grid $G=(V,E)$ with $m=qN^k_{\max}+a$, $n=pN^k_{\max}+b$ and $0\le a,b< N^k_{\max}$, we partition $V$ into the four following sets: 
\[
V_1=\{v_{i,j}|~1\le i\le qN^k_{\max}, 1\le j \le pN^k_{\max} \},
\]
\[ 
V_2=\{v_{i,j}|~qN^k_{\max}+1\le i\le m, 1\le j \le pN^k_{\max} \},
\]
\[ 
V_3=\{v_{i,j}|~1\le i\le qN^k_{\max}, pN^k_{\max}+1\le j \le n \},
\] 
and
\[ 
V_4=\{v_{i,j}|~qN^k_{\max}+1\le i\le m, pN^k_{\max}+1\le j \le n\}.
\] 
As stated, $|V_1\cap U|\le pqN^k_{\max}$. Grid $V_2$ has $a$ columns each having $N^k_{\max}p$ vertices, hence $|V_2\cap U|\le pa$. Similarly, we have $|V_3\cap U|\le qb$. In summary, we so far have $|(V_1\cup V_2\cup V_3)\cap U|\le pqN^k_{\max}+qb+pa$. Note that $pqN^k_{\max}+qb+pa=\frac{mn}{N^k_{\max}}-\frac{ab}{N^k_{\max}}$.

It remains to upper-bound $|V_4\cap U|$. Without loss of generality assume that $a\le b$. Since the number of rows and columns in $V_4$ are less than $N^k_{\max}$, in each row/column at most one dominating vertex can exist. Since $a\le b$, then $|V_4\cap U|\le a$. Therefore $|V\cap U|=|U|\le \frac{mn}{N^k_{\max}}-\frac{ab}{N^k_{\max}}+a$. Maximum of $-\frac{ab}{N^k_{\max}}+a$ takes place when $b$ has its minimum value, i.e., $b=a$. Moreover, for $-\frac{a^2}{N^k_{\max}}+a$ we have that the maximum is $\frac{N^k_{\max}}{4}$ and it happens when $a=\frac{N^k_{\max}}{2}$. This results in $|U|\le \left\lceil \frac{mn}{ N^k_{\max}}+\frac{N^k_{\max}}{4}\right\rceil$.

\end{proof}

\begin{theorem}
\label{theo:upper-bound-k-dist}
For an $m\times n$ grid $G=(V,E)$, a $k$-distance dominating set $S\subset V$ can be constructed using $k$-diagonalization and $k$-projection in polynomial-time such that $|S|\le \left\lceil \frac{(m+2k)(n+2k)}{N^k_{\max}}+\frac{N^k_{\max}}{4}\right\rceil$.
\end{theorem}

\begin{proof}
The proof follows from Lemmas \ref{lem:k-distance-unit}, \ref{lem:k-distane-k-super} and \ref{lem:k-diag-cardinal} and by replacing the diagonalization and projection operations with the $k$-diagonalization and $k$-projection operations in the proof of Theorem \ref{theo:upper-bound}~\cite{TYC'92}.
\end{proof}

\begin{lemma}
\label{lem:lower-bound-k-dist}
If $S\subset V$ is a $k$-distance dominating set for an $m\times n$ grid $G=(V,E)$, $|S|\ge\left\lceil \frac{mn}{N^k_{\max}}\right\rceil$
\end{lemma}

\begin{proof}
According to Lemma \ref{lem:k-distance-unit}, a vertex $v\in V$ $k$-dominates at most $N^k_{\max}$ vertices.
Hence, at least $\left\lceil \frac{mn}{N^k_{\max}}\right\rceil$ dominating vertices are needed to $k$-dominate an $m\times n$ grid. Note that we use $|S|\ge\left\lceil \frac{mn}{N^k_{\max}}\right\rceil$ instead of $|S|\ge\left\lfloor \frac{mn}{N^k_{\max}}\right\rfloor$ since dominating vertices in the $k$-neighbourhood of vertices on the grid boundary do not have all their $k$-neighbours in $V$.
\end{proof}

\begin{corollary}
\label{cor:upper-lower}
Let $S$ be a $k$-distance dominating set for an $m\times n$ grid $G=(V,E)$ obtained by $k$-diagonalization and $k$-projection and let $L$ denote the lower-bound for $S$ from Lemma \ref{lem:lower-bound-k-dist}. For any constant $k\in \mathbb{Z_+}$, the approximation ratio $\frac{|S|}{L}$ satisfies $\lim_{n,m\rightarrow\infty} \frac{|S|}{L}=1$.
\end{corollary}

\begin{proof}
From Theorem \ref{theo:upper-bound-k-dist} and Lemma
\ref{lem:lower-bound-k-dist}, we have
\[
\frac{|S|}{L}\le\frac{\left\lceil(m+2k)(n+2k)/N^k_{\max}+N^k_{\max}/4\right\rceil}{\left\lceil
    mn/N^k_{\max}\right\rceil}.
\]
Therefore,
\[
\frac{(m+2k)(n+2k)/N^k_{\max}+N^k_{\max}/4}{mn/N^k_{\max}+1}\le\frac{|S|}{L},
\]
and
\[
 \frac{|S|}{L}\le\frac{(m+2k)(n+2k)/N^k_{\max}+N^k_{\max}/4+1}{mn/N^k_{\max}}.
\]
Hence, we have
\[ 
\frac{(m+2k)(n+2k)+(N^k_{\max})^2/4}{mn+N^k_{\max}}\le\frac{|S|}{L},
\]
and 
\[
\frac{|S|}{L}\le\frac{(m+2k)(n+2k)+(N^k_{\max})^2/4+N^k_{\max}}{mn}.
\]
For constant $k$ we have
\[
\lim_{n,m\rightarrow\infty}\frac{(m+2k)(n+2k)+(N^k_{\max})^2/4}{mn+N^k_{\max}}=
\]
\[
\lim_{n,m\rightarrow\infty}\frac{(m+2k)(n+2k)+(N^k_{\max})^2/4+N^k_{\max}}{mn}=1.
\]
Therefore by the Squeeze Theorem $\lim_{n,m\rightarrow\infty} \frac{|S|}{L}=1$. \end{proof}

\begin{figure}[t]
\includegraphics[scale=0.2]{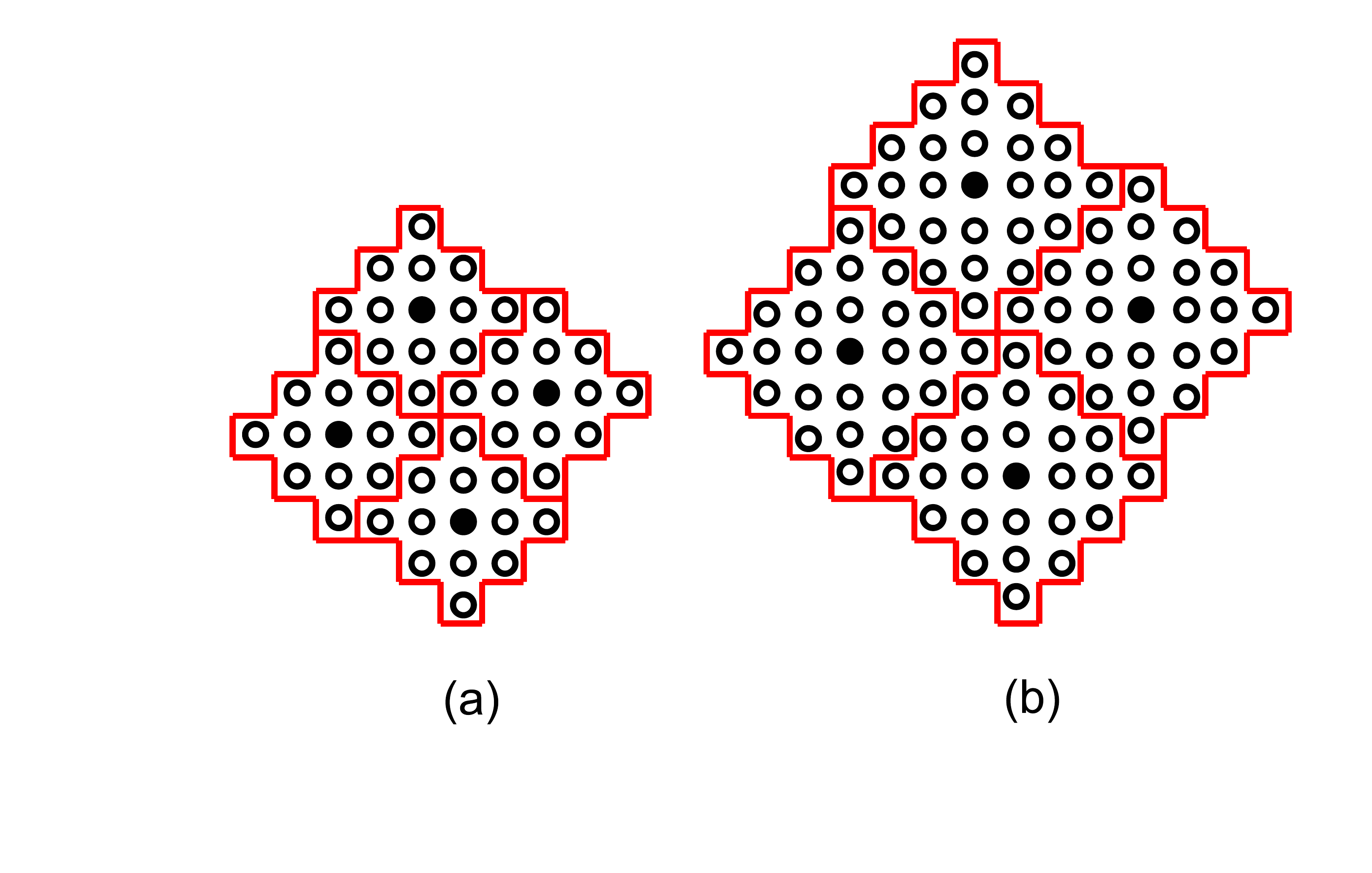}
\vspace{-0.5cm}
\caption{A 2-diagonal pattern and a 3-diagonal pattern are depicted. Observe that the structure is similar to the regular diagonal pattern.}
\label{fig:k-distance-elem}
\end{figure}

For a graph $G$, its \emph{$k$-th power}, denoted by $G^k=(V',E')$, is a graph with the same vertex set as $G$, i.e., $V=V'$, in which two distinct vertices share an edge if and only if their distance in $G$ is at most $k$~\cite{BM'08} (see Figure~\ref{fig:k-power}). Hence, in $G^k$ each vertex is connected to the vertices it $k$-distance dominates in $G$. We finish this section with the following remark that relates the $k$-distance dominating set problem in grids to the regular dominating set problem in their $k$-th power graphs.

\begin{remark}[$k$-th Power of Grids]
\label{rem:k-power}
It might seem that a reasonable approach for $k$-distance domination on a grid $G$ is to simply take the $k$-th power of the graph to obtain $G^k$, and then perform regular domination algorithms on $G^k$. Note that by the definition of $G^k$, a regular dominating set in $G^k$ is equivalent to a $k$-distance dominating set in $G$ and hence $\gamma(G^k)=\gamma^k(G)$. Unfortunately, $G^k$ is no longer a grid (e.g., there are diagonal edges connecting $v_{x,y}$ to $v_{x+1,y+1}$ for $k\ge 2$). In fact, it is not even a planar graph for $m\times n$ grids with $m,n\ge 2$. Therefore, as discussed in Section~\ref{sec:intro}, choosing dominating vertices greedily in $G^k$ might obtain a dominating set with size as large as $(\ln(|V|)+1)\gamma^k(G)$.
\end{remark}

\begin{figure}[t]
\includegraphics[scale=0.4]{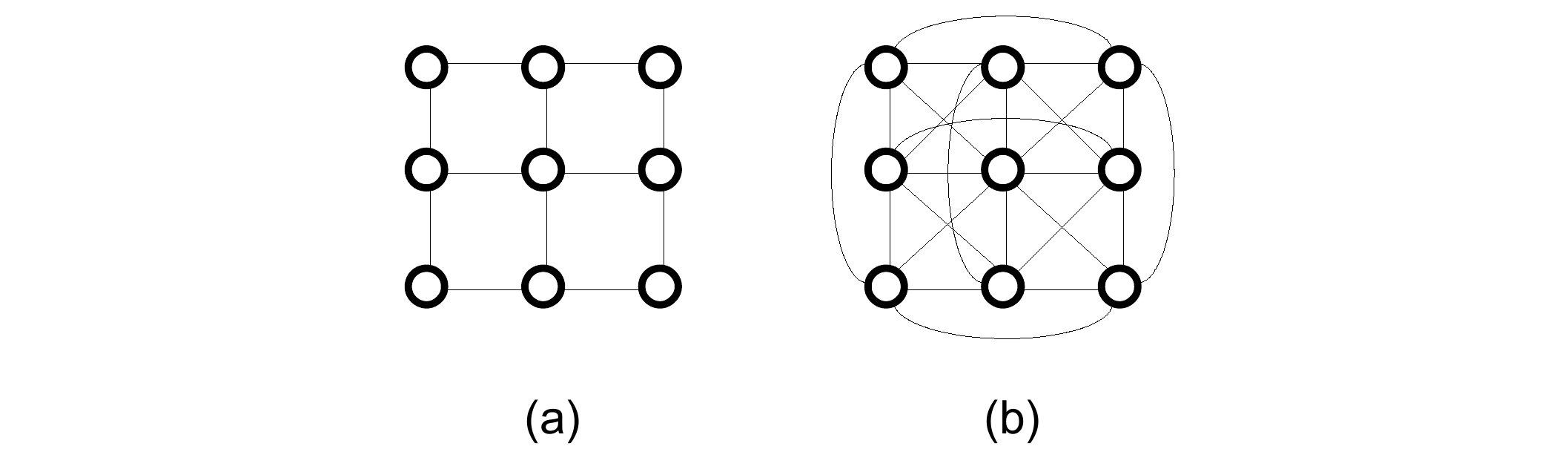}
\caption{Figures (a) and (b) show a $3\times 3$ graph and its second power, respectively. Vertices within distance two are connected to each other in (b).}
\label{fig:k-power}
\end{figure}

\subsection{Distributed $k$-Distance Domination on Grids}
\label{subsection:k-distance-dist}
Using the algorithm explained in Section \ref{subsec:upper-bound-k-distance}, a distributed $k$-distance domination algorithm can be designed for grids. In practice, the $k$-distance dominating set problem corresponds to settings where agents are equipped with longer range sensory equipment and can sense vertices up to distance $k$ from them.
Therefore, the goal is to arrange the agents on the grid vertices in a distributed way such that for each vertex there exists at least one agent with distance at most $k$ from it.

This algorithm is similar to Algorithm \ref{alg:dist-domin} in Section \ref{susec:dist-alg}, except for two modifications. The first modification is that module $m_2$ and module $m_1$ with module centers $c(m_1)=v_{i,j}$ and $c(m_2)=v_{i',j'}$ can now connect to each other if $v_{i',j'}\in \{ v_{i+k,j+k+1},v_{i+k+1,j-k},v_{i-k,j-k-1},v_{i-k-1,j+k}\}$ (see Figure \ref{fig:k-distance-elem}). These constitute the \emph{slots}. The second modification is the definition of \emph{orphans}. If $U'$ is a set of vertices that $k$-diagonalizes the $k$-super-grid of $G$, vertex $v\in V$ is an orphan if it satisfies the two following conditions: (a) $v$ has no $k$-neighbour in $U' \cap V$, and (b) $v$ is in the $k$-neighbourhood of a vertex $u\in U'\backslash V$ with the same $x$ or $y$ coordinates. Hence, \emph{valid slots} are defined for each settled agent as the union of its slots located inside the grid and the orphans of its slots located outside the grid (see Figure \ref{fig:k-grid}).

\begin{figure}[t]
\includegraphics[scale=0.22]{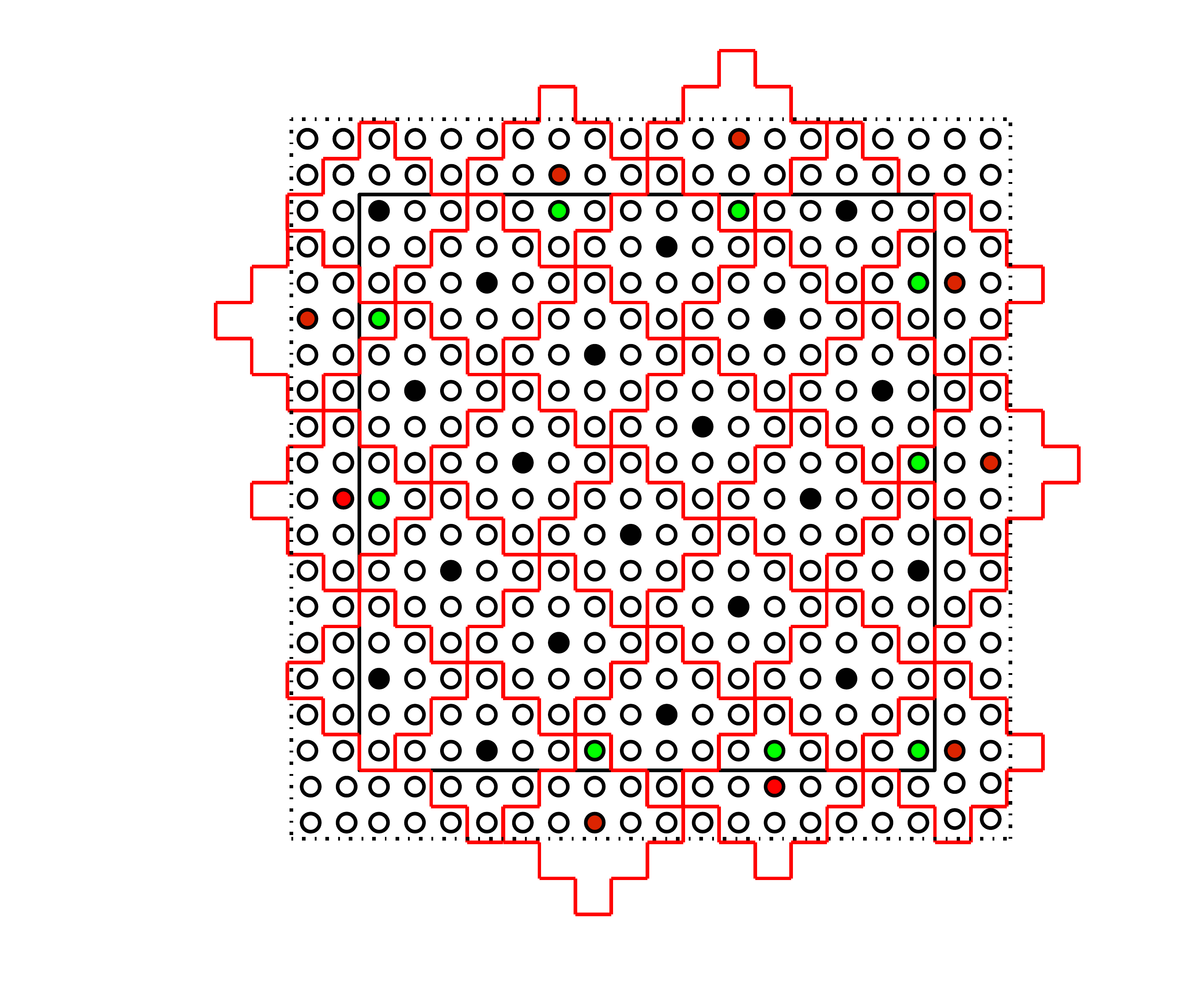}
\vspace{-0.5cm}
\caption{A $16\times 16$ grid $G$ and its 2-super-grid $G'$ are shown by solid and dashed squares, respectively. Both grids are 2-diagonalized. The black circles are the vertices that 2-diagonalize $G$. The union of red and black circles 2-diagonalizes $G'$. The green circles are the 2-projections of the red circles onto $G$. Before projection these vertices are called orphans.}
\label{fig:k-grid}
\end{figure}

\section{Summary and Open problems}
\label{sec:conclusion}
In this paper we studied the dominating set and $k$-distance dominating set problems on $m\times n$ grids. We discussed a construction from~\cite{TYC'92} to obtain dominating sets for grids with near optimal size and generalized it to work in the $k$-distance domination scenario. We used these methods in distributed algorithms and showed that the resulting dominating sets are upper-bounded by $\left\lceil\frac{(m+2k)(n+2k)}{2k^2+2k+1}+\frac{2k^2+2k+1}{4}\right\rceil$. The difference between the acquired upper-bound and the domination number of grid is at most five, for $16\le m\le n$ and $k=1$. However, via a more detailed case-based analysis in the grid corners, our distributed procedure can be used to obtain optimal dominating sets for $16\le m\le n$.

There are many open problems in this area. The $k$-domination number of grids is still unknown. It is also of interest to find centralized and distributed algorithms for dominating sub-graphs of grids, that is, grids with some of their vertices or edges missing. Generalizing these algorithms to the cases where the underlying graphs are cubic or hyper-cubic grids is another direction of this research.

\end{document}